\documentclass[10pt]{amsart}
\usepackage{amssymb}
\usepackage{enumerate}
\usepackage{epsf}
\usepackage{psfrag}
\usepackage{graphics}
\usepackage[dvips]{graphicx}
\usepackage{xcolor}
\usepackage{hyperref}
\usepackage{mathrsfs}
\usepackage{graphicx}
\usepackage{wrapfig}
\usepackage{float}
\DeclareGraphicsExtensions{.eps,.art,.ART,.ps}
\usepackage{mathtools}
\usepackage{graphicx,caption}
\usepackage{subfig}
\newcommand{\cF}{\mathcal{F}}

\newcommand{\bbH}{\mathbb{H}}
\newcommand{\bbS}{\mathbb{S}}

\newcommand{\bbR}{\mathbb{R}}      
\newcommand{\bbN}{\mathbb{N}}      

\newcommand{\grad}{\operatorname{grad}}

\newcommand{\rig}{R}
\newcommand{\lef}{L}

\newtheorem{Thm}{Theorem}[section]

\newtheorem{Remark}[Thm]{Remark}

\begin{document}
\begin{abstract}
We study the motion of a rigid elastic rod, initially set in its relaxed state along a spacelike geodesic, in an expanding Friedmann-Lema\^itre-Robertson-Walker universe. This leads to an initial boundary value problem (IBVP) for a nonlinear wave equation whose nonlinearity depends on a parameter $\kappa \geq 0$, related to the ratio between the rod's length and the cosmological scale. We show that if $\kappa$ is small enough then the solution to the IBVP is global in time and bounded, meaning that the rod's length oscillates around its initial value. For greater values of  $\kappa$, however, the solution to the IBVP blows up in finite coordinate time, indicating that the rod is infinitely stretched by the cosmological expansion. This supports the widely held belief that sufficiently small bound systems do not follow the Hubble flow, whereas larger systems may do so. Similar conclusions apply to the tethered galaxy version of this problem, where the rod is used to connect two point masses (which results in nonlinear boundary conditions for the IBVP).
\end{abstract}
%
%
\title{Elastic rigid rod in an expanding universe}
\author{Sim\~{a}o Correia, Jos\'{e} Nat\'{a}rio and Jorge Drumond Silva}
\address{Center for Mathematical Analysis, Geometry and Dynamical Systems, Mathematics Department, Instituto Superior T\'ecnico, Universidade de Lisboa, Portugal}
%
%
\maketitle
%
%
%
%
\thispagestyle{empty}
%
%
%
\section*{Introduction}
Nearly a century after Hubble's foundational discovery of the expansion of the Universe \cite{Hubble29}, the question of what kind of extended physical systems undergo this expansion remains a subject of ongoing scientific debate. The prevailing view seems to be that sufficiently small bound systems do not follow the Hubble flow (see \cite{McVittie33, ES45, DP64, NP71, Anderson95, Bonnor96, CFV98, Bonnor99, FJ07, CG10, NLH12, PR12} and references therein).

Relativistic elasticity (see, for instance, \cite{BS03, KS04, Beig23}) provides a tractable framework for modeling extended systems bound by elastic forces. In this work, we investigate the behavior of an elastic body embedded in an expanding Friedmann-Lemaître-Robertson-Walker (FLRW) universe. For simplicity, we consider a rigid elastic rod of the type studied in \cite{Natario14, NQV18, CN19, MNS24}, initially placed in its relaxed configuration along a spacelike geodesic.
This setup gives rise to an initial boundary value problem (IBVP) for a nonlinear wave equation, where the nonlinearity depends on a dimensionless parameter $\kappa \geq 0$, representing the square of the ratio between the rod's proper length and the characteristic cosmological scale. We show that when $\kappa$ is sufficiently small the solution to the IBVP is global in time and bounded, meaning that the rod's length oscillates around its initial value. For larger values of $\kappa$, however, the solution blows up in finite coordinate time, which indicates that the rod is stretched without bound by the cosmological expansion.
Similar results are obtained in the ``tethered galaxy'' version of the problem, where the elastic rod connects two point masses. Although this variation gives rise to nonlinear boundary conditions in the IBVP, which require a more intricate analysis to establish an existence and uniqueness result, the qualitative behavior of the solutions remains essentially unchanged.

The organization of this paper is as follows. In Section~\ref{section1}, we briefly review key properties of the FLRW metrics that will be considered in this work. Section~\ref{section2} provides a concise overview of the general theory of one-dimensional rigid elastic bodies evolving in two-dimensional spacetimes, which suffices to study the motion of the rod. In Section~\ref{section3}, we formulate the IBVP and prove that its solutions exist globally in time for sufficiently small values of the parameter $\kappa$, while solutions exhibit finite-time blow-up for large enough values of $\kappa$. The special case of FLRW models with constant curvature, corresponding to the de Sitter universe, is discussed in Section~\ref{section4}. Finally, in Section~\ref{section5}, we extend our analysis to the tethered galaxy problem.

We mostly follow the conventions of \cite{MTW73}; in particular, we use a geometrized system of units, in which both the speed of light $c$ and Newton's gravitational constant $G$ are set to unity, $c=G=1$.
%
%
%
\section{FLRW universes}\label{section1}
Arguably, the simplest models of expanding universes are given by the FLRW metrics,
\begin{equation} \label{FLRWmetric}
ds^2 = - dt^2 + a^2(t) h_{ij} dx^i dx^j \, .
\end{equation}
These metrics are defined on $\bbR^+ \times \Sigma$, where $(\Sigma, h)$ is the simply connected Riemannian manifold of constant curvature $0$, $1$ or $-1$, that is, $\bbR^n$, $\bbS^n$ or $\bbH^n$ with the standard metric. The lines of constant spatial coordinates $(x^1,x^2,x^3)$ are timelike geodesics, typically interpreted as representing galaxies moving away from each other; this motion, known as the Hubble flow, is governed by the scale factor function $a(t) > 0$, which we assume to be strictly increasing and unbounded. If the galaxies are modeled as a pressureless fluid of uniform density $\rho=\rho(t)$ then the Einstein equations for \eqref{FLRWmetric} imply
\begin{equation} \label{Friedmann}
\frac{\ddot{a}}{a} = - \frac{4\pi\rho}{3} + \frac{\Lambda}{3} \, , \qquad \frac{d}{dt} \left( \rho a^3 \right) = 0 \, ,
\end{equation}
where $\Lambda$ is the  cosmological constant. Notice that our hypotheses on $a(t)$ require that $\Lambda \geq 0$.

If $\gamma \subset \Sigma$ is the image of a geodesic then $\bbR^+ \times \gamma$ is a totally geodesic submanifold of $\bbR^+ \times \Sigma$, whose induced metric can be put in the form
\begin{equation} \label{FLRWmetric2D}
ds^2 = - dt^2 + a^2(t) dx^2 \, ,
\end{equation}
with Gauss curvature
\begin{equation}
K = \frac{\ddot{a}}{a} = - \frac{4\pi\rho}{3} + \frac{\Lambda}{3}
\end{equation}
(see Appendix~\ref{appendixA}). It follows from equations \eqref{Friedmann} and our hypotheses on $a(t)$ that $K(t)$ is a nondecreasing function with
\begin{equation}
\lim_{t \to \infty} K(t) = \frac{\Lambda}3 \geq 0 \, .
\end{equation}

Now consider a group of galaxies initially arranged along a geodesic segment of $\Sigma$. If, instead of allowing them to evolve as non-interacting test particles, we bind them together into an elastic medium, then we can expect their motion to deviate from the Hubble flow. Nevertheless, tidal forces arising from the nonvanishing curvature may try to pull the galaxies apart, straining the elastic medium. Whether this effect is sufficient to overcome the elastic restoring forces depends on the details of the elastic model, which we discuss in the next section.
%
%
%
\section{Elastic rod}\label{section2}
Since we are interested in the motion of a one-dimensional elastic body (an elastic rod) along a spatial geodesic direction, we may restrict ourselves to a two-dimensional spacetime $(M,g)$ (namely the totally geodesic two-dimensional spacetime submanifold along which the motion takes place). The elastic rod can then be modeled by a scalar field $\lambda : M \to \mathbb{R}$ with nonvanishing spacelike gradient, whose level sets correspond to the worldlines of the particles making up the rod (see \cite{Natario14} for more details). If we choose $\lambda$ to be the proper length along the undeformed rod, then the local compression factor is given by the positive function $n$ defined via
\begin{equation}
n^2 = {|\grad \lambda|}^2 = ({\nabla}^{\alpha}\lambda)({\nabla}_{\alpha}\lambda) \, .
\end{equation}

Because we are considering a two-dimensional spacetime, the energy-momentum tensor of the elastic rod necessarily takes the form of a perfect fluid:
\begin{equation}
T_{\alpha\beta} = (\mu + p) u_\alpha u_\beta + p g_{\alpha\beta} \, ,
\end{equation}
where $\mu$ is the energy density (per unit length), $p$ is the pressure (actually a force in this two-dimensional setting) and $u$ is the unit tangent covector to the rod particle's worldlines.
For simplicity, we assume that the rod is rigid, meaning that its speed of sound equals the speed of light:
\begin{equation}
\frac{dp}{d\mu} = 1 \Rightarrow \mu = \mu_0 + p \, ,
\end{equation}
where $\mu_0$ is the (constant) density in the relaxed state. Under this assumption, as shown in \cite{Natario14}, the energy density and the pressure are given by
\begin{equation}
\mu = \frac{\mu_0}2 \left( n^2 + 1 \right) \, , \qquad p = \frac{\mu_0}2 \left( n^2 - 1 \right) \, ,
\end{equation}
and so the rod's energy-momentum tensor becomes
\begin{equation}\label{eq:energymomentumtensor}
T_{\alpha\beta} = \mu_0 \left[ {\nabla}_{\alpha} \lambda \, {\nabla}_{\beta}\lambda - \frac12 ({\nabla}^{\mu}\lambda\,{\nabla}_{\mu}\lambda) g_{\alpha\beta} - \frac12 g_{\alpha\beta} \right] \, ,
\end{equation}
which is, up to a term with zero divergence, the energy-momentum tensor of a massless scalar field $\lambda$. Consequently, one can readily show that $\lambda$ must satisfy the wave equation:
\begin{equation}
\Box \lambda = 0 \, .
\end{equation}

Generically, the motion of a finite rod leads to a free boundary problem with nonlinear boun\-da\-ry conditions for this equation.
Such problems are highly nontrivial, and a judicious choice of comoving coordinates, for which where the boundary remains fixed, is typically the most effective approach to study them. In this simple two-dimensional setting, a natural comoving coordinate is the function $\lambda$ itself, since the free boundary, corresponding to the endpoints of the rod, is given by level sets of this function. As shown in \cite{Natario14}, one can always choose a conjugate time function $\tau:M \to \bbR$, also satisfying the wave equation, such that the metric in the (harmonic) coordinates $(\tau,\lambda)$ is written as
\begin{equation} \label{proper}
ds^2 = \frac{1}{n^2(\tau,\lambda)} \left( - d\tau^2 + d\lambda^2 \right) \, .
\end{equation}
Conversely, if the metric can be expressed in the form \eqref{proper} then $\lambda$ is a solution of the wave equation. Therefore, a method for solving the problem of the rod's motion is simply to find a suitable harmonic coordinate system.
%
%
\section{Initial boundary value problem}\label{section3}
To solve the free boundary problem we perform a change of variable, setting
\begin{equation}
n(\tau,\lambda) = e^{u(\tau,\lambda)} \, ,
\end{equation}
so that the Gauss curvature of the two-dimensional metric~\eqref{proper} is given by
\begin{equation}
K = e^{2u} \left( - \ddot{u} + u'' \right) \, ,
\end{equation}
where $\dot{} \equiv \frac{\partial}{\partial \tau}$ and ${}' \equiv \frac{\partial}{\partial \lambda}$ (see Appendix~\ref{appendixA}). We assume that the curvature function satisfies the monotonicity condition
\begin{equation}
\dot{K} \geq 0 \, ,
\end{equation}
which holds for the expanding FLRW models discussed in Section~\ref{section1}, as well as for many other expanding cosmological spacetimes.

Assuming that the rod is initially relaxed and free of internal motions, we are led to the following IBVP:
\begin{equation} \label{IBVP}
\begin{cases}
- \ddot{u} + u'' = K e^{- 2u} \\
u(0,\lambda)=\dot{u}(0,\lambda)=0, \quad \lambda \in [0,1] \\
u(\tau,0) = u(\tau,1) = 0, \quad \tau\geq 0
\end{cases} \, ,
\end{equation}
where we normalize the length of the rod to unity.

Consider the energy functional
\begin{equation}
H(\tau) = \int_0^1 \left( \dot{u}^2 + u'^2 - K e^{-2u} \right) d\lambda \, .
\end{equation}
Differentiating with respect to time yields
\begin{align}
\dot{H}(\tau) & = \int_0^1 \left( 2\dot{u}\ddot{u} + 2u'\dot{u}' + 2K e^{-2u}\dot{u} - \dot{K} e^{-2u} \right) d\lambda = \int_0^1 \left( 2\dot{u}u'' + 2u'\dot{u}' - \dot{K} e^{-2u} \right) d\lambda \nonumber \\
& = \Big[ 2\dot{u}u' \Big]_0^1 - \int_0^1 \dot{K} e^{-2u}d\lambda \leq 0 \, . \label{energycalculation}
\end{align}
Hence, $H(\tau)$ is a nonincreasing function of $\tau$. This monotonicity property allows us to establish a global existence and boundedness result, applicable to the expanding FLRW models discussed in Section~\ref{section1} in the case when $\Lambda = 0$.

\begin{Thm}
If $\dot{K} \geq 0$ and $K \leq 0$ then the solution to the IBVP \eqref{IBVP} is global in time towards the future and bounded.
\end{Thm}

\begin{proof}
According to the blow-up alternative in Theorem~\ref{GWPBA} (see Appendix~\ref{appendixB}), if $\|u(\tau,\cdot)\|_{L^\infty}$ remains bounded for $\tau>0$ then the unique solution to the IBVP can be extended globally in time towards the future. This norm is controlled by the $L^2$ norm of the spatial derivative, since
\begin{equation}
u(\tau,\lambda) = \int_0^\lambda u'(\tau,s) ds \Rightarrow \|u(\tau,\cdot)\|_{L^\infty} \leq \|u'(\tau,\cdot)\|_{L^1} \leq \|u'(\tau,\cdot)\|_{L^2} \, ,
\end{equation}
where we used the Cauchy-Schwarz inequality. Consequently,
\begin{equation} \label{Linftybound}
\|u\|^2_{L^\infty} \leq \|u'\|^2_{L^2} \leq H(\tau) + \int_0^1 K e^{-2u} d\lambda \leq H(0) + \int_0^1 K e^{-2u} d\lambda \, .
\end{equation}
If $K \leq 0$ we then have
\begin{equation}
\|u\|^2_{L^\infty} \leq H(0) \, .
\end{equation}
\end{proof}

Regarding the expanding FLRW models discussed in Section~\ref{section1}, this result states that if the cosmological constant vanishes then the rod never follows the Hubble flow. The next result implies that this is still true in the case when $\Lambda > 0$ if $K$ is small enough.

\begin{Thm} \label{globalintime}
If $\dot{K} \geq 0$ and $0 \leq K \leq \kappa < 1/e^2$ then the solution to the IBVP \eqref{IBVP} is global in time towards the future and bounded.
\end{Thm}

\begin{proof}
Since $K \geq 0$, we have
\begin{equation}
H(0) = - \int_0^1 K(0,\lambda) d\lambda \leq 0 \, ,
\end{equation}
and so we obtain from \eqref{Linftybound}
\begin{equation} \label{ineq}
\|u\|^2_{L^\infty} \leq \int_0^1 K e^{-2u} d\lambda \leq \kappa e^{2 \|u\|_{L^\infty}} \, .
\end{equation}
Consider the function
\begin{equation}
f(x) = x^2 - \kappa e^{2x} \, ,
\end{equation}
which satisfies\footnote{It can easily be shown that $f$ becomes positive in $\bbR^+$ if and only if $\kappa < 1/e^2$.}
\begin{equation}
f(0)=-\kappa < 0 \qquad \text{ and } \qquad f(1)=1-\kappa e^2 > 0 \, .
\end{equation}
From \eqref{ineq} we have
\begin{equation}
f(\|u\|_{L^\infty}) \leq 0 \, ,
\end{equation}
and since $\|u(0,\cdot)\|_{L^\infty}=0$ we conclude that $\|u(\tau,\cdot)\|_{L^\infty}<1$ for $\tau>0$.
\end{proof}

The condition $K \geq 0$ corresponds, in the context of the expanding FLRW models discussed in Section~\ref{section1}, to the regime in which the cosmological constant dominates the expansion -- that is, when $\Lambda > 4 \pi \rho$. In this case, $\kappa$ effectively measures the dimensionless ratio between the square of the rod’s length and the square of the asymptotic cosmological radius $R=\sqrt{3/\Lambda}$, and so small $\kappa$ corresponds to a rod whose length is much shorter by the cosmological scale.

When the curvature becomes sufficiently large, on the other hand, we obtain the following finite-time blow-up result.

\begin{Thm} \label{blowupThm}
If $\dot{K} \geq 0$ and $K \geq \kappa > \pi^2$ then the solution to the IBVP \eqref{IBVP} blows up before $\tau=1/2$.
\end{Thm}

\begin{proof}
The proof consists in showing that under these hypotheses there is an ODE mechanism driving the blow-up. Consider the future Cauchy development of the initial data for the IBVP \eqref{IBVP},
\begin{equation}
\mathcal{D} = \left\{ (\tau,\lambda) \in \bbR^2 : 0 \leq \tau \leq \frac12 - \left|\lambda-\frac12\right| \right\}
\end{equation}
(see Figure~\ref{Domain}), and the vector fields

\begin{figure}[h!]
\begin{center}
\psfrag{t}{$\tau$}
\psfrag{l}{$\lambda$}
\psfrag{0}{$0$}
\psfrag{1}{$1$}
\psfrag{1/2}{$\frac12$}
\psfrag{D+}{$\mathcal{D}$}
\epsfxsize=.5\textwidth
\leavevmode
\epsfbox{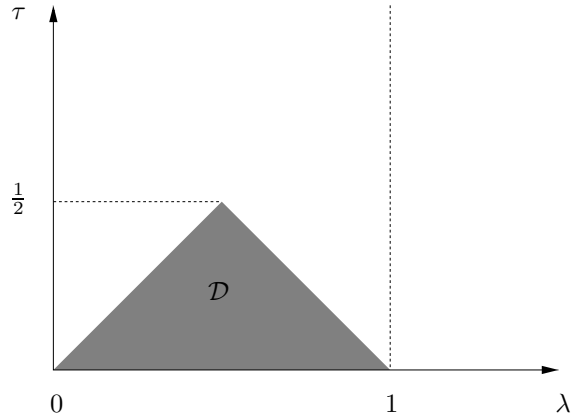}
\end{center}
\caption{Future Cauchy development of the initial data for the IBVP \eqref{IBVP}.} \label{Domain}
\end{figure}

\begin{equation} \label{L+-}
L^- = \frac{\partial}{\partial \tau} - \frac{\partial}{\partial \lambda} \qquad \text{ and } \qquad L^+ = \frac{\partial}{\partial \tau} + \frac{\partial}{\partial \lambda} \, ,
\end{equation}
whose integral curves are of the form
\begin{equation}
\tau \mapsto (\tau, \lambda_0-\tau) \qquad \text{ and } \qquad \tau \mapsto (\tau, \lambda_0+\tau)
\end{equation}
(with $\lambda_0$ constant). The nonlinear wave equation in \eqref{IBVP} can be written as
\begin{equation} 
L^- L^+ u = - K e^{-2u} \, ,
\end{equation}
whence
\begin{equation} \label{wavenullineq}
L^- L^+ u \leq - \kappa e^{-2u} \, .
\end{equation}
From the initial conditions in \eqref{IBVP} it is clear that
\begin{equation}
L^+ u(0,\lambda) = 0
\end{equation}
for $\lambda \in [0,1]$, and so \eqref{wavenullineq} implies
\begin{equation}
L^+ u \leq 0 \qquad \text{on } \mathcal{D} \, ,
\end{equation}
whence
\begin{equation}
u \leq 0 \qquad \text{on } \mathcal{D} \, .
\end{equation}
This inequality can be used on \eqref{wavenullineq} to obtain the improved estimate
\begin{equation}
L^+ u \leq - \kappa \tau \qquad \text{on } \mathcal{D} \, ,
\end{equation}
whence
\begin{equation}
u \leq - \frac12 \kappa \tau^2 \qquad \text{on } \mathcal{D} \, .
\end{equation}
By iterating the procedure above we can construct a sequence of estimates
\begin{equation} \label{estimates}
u \leq v_n(\tau) \, , \qquad L^+ u \leq w_n(\tau) \qquad \text{on } \mathcal{D} \, ,
\end{equation}
where
\begin{equation}
v_0(\tau) = w_0(\tau) = 0 \, , \qquad v_{n+1}(\tau) = \int_0^\tau w_n(s) ds \, , \qquad w_{n+1}(\tau) = - \kappa \int_0^\tau e^{-2v_n(s)} ds \, .
\end{equation}
It is easy to show (by induction) that the sequences $v_n(\tau)$ and $w_n(\tau)$ are decreasing for each $\tau \in [0,\frac12)$, and so they have limits $v(\tau)$ and $w(\tau)$ (possibly $-\infty$). By the Monotone Convergence Theorem, the limit functions $v(\tau)$ and $w(\tau)$ satisfy
\begin{equation} \label{limitfunctions}
v(\tau) = \int_0^\tau w(s) ds \, , \qquad w(\tau) = - \kappa \int_0^\tau e^{-2v(s)} ds \, .
\end{equation}
Note that $v_n(\tau)$ and $w_n(\tau)$ are clearly decreasing functions of $\tau$, and consequently so are the limit functions $v(\tau)$ and $w(\tau)$. In particular, if $v(\tau_0)=-\infty$ then $v(\tau)=-\infty$ for all $\tau> \tau_0$ (and similarly for $w(\tau)$). From \eqref{limitfunctions} we then see that $v(\tau)$ and $w(\tau)$ are smooth real-valued functions on some interval of the form $[0,\tau_*)$, and identically $-\infty$ on $(\tau_*,\frac12)$. Moreover, on  $[0,\tau_*)$ the function $v(\tau)$ satisfies the ODE Cauchy problem
\begin{equation}
\begin{cases}
\ddot{v} = - \kappa e^{- 2v} \\
v(0)=\dot{v}(0)=0
\end{cases} \, ,
\end{equation}
whose solution is
\begin{equation}
v(\tau) = \log \cos \left( \sqrt{\kappa} \tau \right) \, .
\end{equation}
This shows that
\begin{equation}
\tau_* = \frac{\pi}{2\sqrt{\kappa}} < \frac12 \, .
\end{equation}
Since from \eqref{estimates} we have
\begin{equation}
u \leq v(\tau) \qquad \text{on } \mathcal{D} \, ,
\end{equation}
we conclude that $u$ blows up before $\tau = 1/2$.
\end{proof}

Concerning the expanding FLRW models discussed in Section~\ref{section1}, this result states that if the initial length of the rod is comparable to the cosmological scale then the rod is infinitely stretched by the Hubble flow.
%
%
\section{The constant curvature case}\label{section4}
In the constant curvature case, where the FLRW model is the de Sitter universe with cosmological radius $R = 1/\sqrt{K}$, there is a simple geometric interpretation of the blowup in Theorem~\ref{blowupThm}: the condition $K > \pi^2$ is simply the condition that $R < 1/\pi$, meaning that the diameter of the region bounded by the cosmological horizon (measured along a spacelike geodesic) is $\pi R < 1$. Since we have set the rod's length to $1$, this means that the rod initially extends beyond the cosmological horizon, and so the future Cauchy development of the initial data for the rod must contain a segment of future null infinity $\mathscr{I^+}$ (see Figure~\ref{Pen_dS}). This means that the length of the rod cannot remain bounded, and in fact blows up in finite coordinate time $\tau$.

\begin{figure}[h!]
\begin{center}
\psfrag{I+}{$\mathscr{I^+}$}
\psfrag{I-}{$\mathscr{I^-}$}
\epsfxsize=.5\textwidth
\leavevmode
\epsfbox{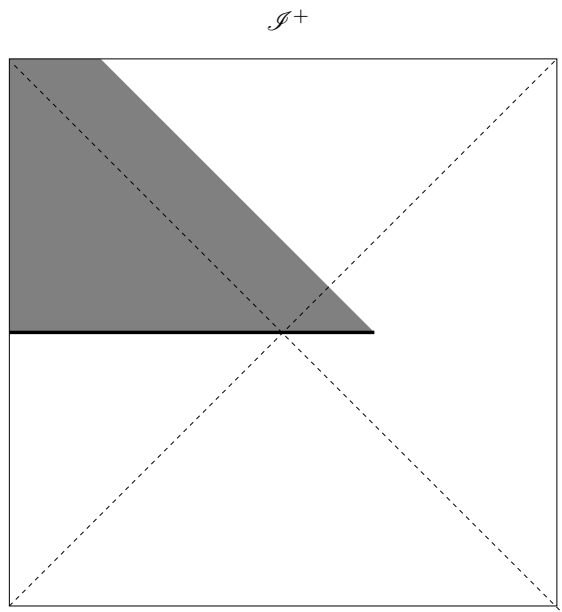}
\end{center}
\caption{Penrose diagram for the de Sitter spacetime, depicting the initial configuration of the rod centered at $r=0$ (thick line segment, whose left endpoint represents the rod's midpoint) and its future Cauchy development (shaded region). In the case $K > \pi^2$, the rod initially extends beyond the cosmological horizon, and so the future Cauchy development of its initial data must contain a segment of $\mathscr{I^+}$.} \label{Pen_dS}
\end{figure}

The next result demonstrates that the rod’s length still blows up in scenarios where the geometric interpretation above no longer applies.

\begin{Thm} \label{constantcurvature}
If $K>0$ is constant then the solution to the IBVP \eqref{IBVP} blows up for some $K < \pi^2$.
\end{Thm}

\begin{proof}
Assume that $0<K<\pi^2$ is constant. It is easy to check that
\begin{equation}
u(\tau,\lambda) = \log \cos \left( \sqrt{K} \tau \right) \qquad \text{on } \mathcal{D} \, ,
\end{equation}
whence
\begin{equation}
L^+u(\tau,\lambda) = - \sqrt{K} \tan \left( \sqrt{K} \tau \right) \qquad \text{on } \mathcal{D} \, .
\end{equation}
From
\begin{equation} \label{main2}
L^- L^+ u = - K e^{-2u} \leq 0
\end{equation}
and the fact that $\tau + \lambda$ is constant along the integral lines of $L^-$ we obtain
\begin{equation} 
L^+ u \leq - \sqrt{K} \tan \left( \sqrt{K} \, \frac{\tau + \lambda}2 \right) \qquad \text{on } \mathcal{B} \, ,
\end{equation}
where the region $\mathcal{B}$ is defined as
\begin{equation}
\mathcal{B} = \left\{ (\tau,\lambda) \in \bbR^2 : 0 \leq \lambda \leq \frac12 - \left|\tau-\frac12\right| \right\}
\end{equation}
(see Figure~\ref{Domain2}). 

\begin{figure}[h!]
\begin{center}
\psfrag{t}{$\tau$}
\psfrag{l}{$\lambda$}
\psfrag{0}{$0$}
\psfrag{1}{$1$}
\psfrag{1/2}{$\frac12$}
\psfrag{B}{$\mathcal{B}$}
\epsfxsize=.5\textwidth
\leavevmode
\epsfbox{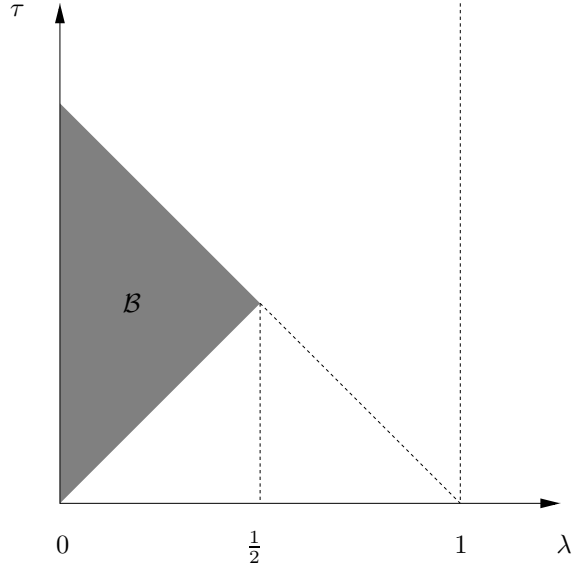}
\end{center}
\caption{Region $\mathcal{B}$.} \label{Domain2}
\end{figure}

This can be rewritten as
\begin{equation} 
\frac{d}{d\tau} u(\tau, \lambda_0 + \tau) \leq - \sqrt{K} \tan \left( \sqrt{K} \, \frac{\lambda_0 + 2\tau}2 \right) \, ,
\end{equation}
and so, integrating from $\lambda=0$, we obtain
\begin{equation} 
u(\tau, \lambda_0 + \tau) \leq - \int_{-\lambda_0}^\tau \sqrt{K} \tan \left( \sqrt{K} \, \frac{\lambda_0 + 2s}2 \right) ds = \log \frac{\cos \left( \sqrt{K} \, \frac{\lambda_0 + 2\tau}2 \right)}{\cos \left( \sqrt{K} \, \frac{\lambda_0}2 \right)} \, ,
\end{equation}
that is,
\begin{equation} 
u \leq \log \frac{\cos \left( \sqrt{K} \, \frac{\tau + \lambda}2 \right)}{\cos \left( \sqrt{K} \, \frac{\tau - \lambda}2 \right)} \qquad \text{on } \mathcal{B} \, .
\end{equation}
From \eqref{main2} we then have
\begin{equation} 
L^- L^+ u \leq - K \frac{\cos^2 \left( \sqrt{K} \, \frac{\tau - \lambda}2 \right)}{\cos^2 \left( \sqrt{K} \, \frac{\tau + \lambda}2 \right)} \qquad \text{on } \mathcal{B} \, ,
\end{equation}
which can be rewritten as
\begin{equation} 
\frac{d}{d\tau} L^+u(\tau, \lambda_0 - \tau) \leq - K \frac{\cos^2 \left( \sqrt{K} \, \frac{2\tau - \lambda_0}2 \right)}{\cos^2 \left( \sqrt{K} \, \frac{\lambda_0}2 \right)} \, .
\end{equation}
Integrating from the boundary of $\mathcal{D}$, where $\tau = \lambda$, yields
\begin{equation} 
L^+u(\tau, \lambda_0 - \tau) + \sqrt{K} \tan \left( \sqrt{K} \, \frac{\lambda_0}2 \right) \leq - \int_{\frac{\lambda_0}2}^\tau K \frac{\cos^2 \left( \sqrt{K} \, \frac{2s - \lambda_0}2 \right)}{\cos^2 \left( \sqrt{K} \, \frac{\lambda_0}2 \right)} ds \, .
\end{equation}
Introducing two constants $\delta, \Delta > 0$ such that $\delta + \Delta < \frac12$ (whose meaning will become apparent), and assuming that
\begin{equation} 
\tau - \lambda \leq 2 \delta + \Delta \Leftrightarrow 2\tau - \lambda_0 \leq 2 \delta + \Delta \, ,
\end{equation}
we then have
\begin{equation} 
L^+u(\tau, \lambda_0 - \tau) \leq - \sqrt{K} \tan \left( \sqrt{K} \, \frac{\lambda_0}2 \right) - K \frac{(2\tau - \lambda_0)\cos^2 \left( \sqrt{K} \, \frac{2 \delta + \Delta}2 \right)}{2\cos^2 \left( \sqrt{K} \, \frac{\lambda_0}2 \right)} \, ,
\end{equation}
or, equivalently,
\begin{equation} 
L^+ u \leq - \sqrt{K} \tan \left( \sqrt{K} \, \frac{\tau + \lambda}2 \right) - K \frac{(\tau - \lambda)\cos^2 \left( \sqrt{K} \, \frac{2 \delta + \Delta}2 \right)}{2\cos^2 \left( \sqrt{K} \, \frac{\tau + \lambda}2 \right)} \quad \text{on } \mathcal{B} \cap \left\{ \tau - \lambda \leq 2 \delta + \Delta \right\} \, .
\end{equation}
Consider in particular the line segment
\begin{equation}
\mathcal{L} = \left\{ (\tau,\lambda) \in \bbR^2 : \tau = \frac12 + \delta \text{ and } \frac12 - \delta - \Delta \leq \lambda \leq \frac12 - \delta \right\}
\end{equation}
(see Figure~\ref{Domain3}). We have

\begin{figure}[h!]
\begin{center}
\psfrag{t}{$\tau$}
\psfrag{l}{$\lambda$}
\psfrag{1}{$1$}
\psfrag{1/2-d-D}{$\frac12\!-\!\delta\!-\!\Delta$}
\psfrag{1/2-d}{$\frac12\!-\!\delta$}
\psfrag{1/2+d}{$\frac12+\delta$}
\psfrag{L}{$\mathcal{L}$}
\psfrag{D'}{$\mathcal{D}'$}
\epsfxsize=.5\textwidth
\leavevmode
\epsfbox{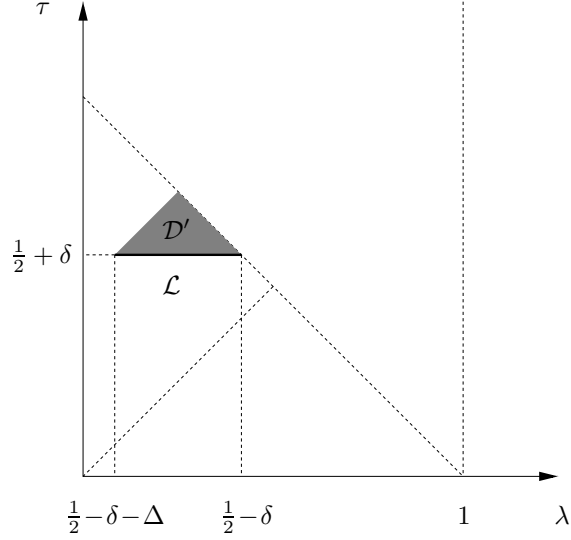}
\end{center}
\caption{Line segment $\mathcal{L}$ and its future Cauchy development.} \label{Domain3}
\end{figure}

\begin{equation} 
u \leq \log \frac{\cos \left( \sqrt{K} \, \frac{1 - \Delta}2 \right)}{\cos \left( \sqrt{K} \, \frac{2\delta + \Delta}2 \right)} \qquad \text{on } \mathcal{L}
\end{equation}
and
\begin{equation} 
L^+ u \leq - \sqrt{K} \tan \left( \sqrt{K} \, \frac{1 - \Delta}2 \right) - \frac{K \delta\cos^2 \left( \sqrt{K} \, \frac{2\delta + \Delta}2 \right)}{\cos^2 \left( \sqrt{K} \, \frac{1 - \Delta}2 \right)} \qquad \text{on } \mathcal{L} \, .
\end{equation}
Let $\mathcal{D}'$ be the future Cauchy development of $\mathcal{L}$ (see Figure~\ref{Domain3}). Arguing as in the proof of Theorem~\ref{blowupThm}, we conclude that
\begin{equation} \label{arguing}
u \leq v(\tau) \, , \qquad L^+ u \leq \dot{v}(\tau) \qquad \text{on } \mathcal{D}' \, ,
\end{equation}
where $v(\tau)$ is the solution of the ODE Cauchy problem
\begin{equation}
\begin{cases}
\ddot{v} = - K e^{- 2v} \\
v(\frac12 + \delta)= \log \frac{\cos \left( \sqrt{K} \, \frac{1 - \Delta}2 \right)}{\cos \left( \sqrt{K} \, \frac{2\delta + \Delta}2 \right)} \\
\dot{v}(\frac12 + \delta)= - \sqrt{K} \tan \left( \sqrt{K} \, \frac{1 - \Delta}2 \right) - \frac{K \delta\cos^2 \left( \sqrt{K} \, \frac{2\delta + \Delta}2 \right)}{\cos^2 \left( \sqrt{K} \, \frac{1 - \Delta}2 \right)}
\end{cases} \, .
\end{equation}
Since the expression for the general solution of the ODE is\footnote{This is the correct expression for the range of initial data $(v_0,\dot{v}_0)$ that that we will be considering, characterized by $\dot{v}_0e^{v_0}<-\sqrt{K}$; for initial data satisfying $\dot{v}_0e^{v_0}>\sqrt{K}$ or $-\sqrt{K}<\dot{v}_0e^{v_0}<\sqrt{K}$ one has the alternative expressions $v(\tau) = \log \left(\frac1{\alpha}\sinh \left( \alpha\sqrt{K} (\tau - \beta) \right)\right)$ or $v(\tau) = \log \left(\frac1{\alpha}\cos \left( \alpha\sqrt{K} (\tau - \beta) \right)\right)$.}
\begin{equation} \label{solutionODE}
v(\tau) = \log \left(\frac1{\alpha}\sinh \left( \alpha\sqrt{K} (\beta - \tau) \right)\right) \, ,
\end{equation}
we must have
\begin{equation} \label{nonlinearsys}
\begin{cases}
\frac1{\alpha}\sinh \left( \alpha\sqrt{K} \frac{2\beta - 1 - 2 \delta}2 \right) = \frac{\cos \left( \sqrt{K} \, \frac{1 - \Delta}2 \right)}{\cos \left( \sqrt{K} \, \frac{2\delta + \Delta}2 \right)} \\
- \alpha \sqrt{K} \coth \left( \alpha\sqrt{K} \frac{1 + 2 \delta - 2\beta}2 \right) = - \sqrt{K} \tan \left( \sqrt{K} \, \frac{1 - \Delta}2 \right) - \frac{K \delta\cos^2 \left( \sqrt{K} \, \frac{2\delta + \Delta}2 \right)}{\cos^2 \left( \sqrt{K} \, \frac{1 - \Delta}2 \right)}
\end{cases} \, .
\end{equation}
Now the solution~\eqref{solutionODE} blows up at time $\tau = \frac12 + \delta + \Delta\tau$, where
\begin{equation}
\Delta\tau = \frac{2\beta - 1 - 2 \delta}2 \, .
\end{equation}
If $\Delta \tau < \frac12 \Delta$ then \eqref{arguing} implies that $u$ blows up in $\mathcal{D}'$. This can be shown to be indeed the case for appropriate values of $(K, \delta, \Delta)$ close to $(\pi^2,0,0)$. In fact, let us rewrite \eqref{nonlinearsys} as
\begin{equation} \label{nonlinearsys2}
\begin{cases} 
\alpha = \frac{\sinh\left( \alpha\sqrt{K} \Delta \tau \right) \cos \left( \sqrt{K} \, \frac{2\delta + \Delta}2 \right)}{\cos \left( \sqrt{K} \, \frac{1 - \Delta}2 \right)} \\
\cosh\left( \alpha\sqrt{K} \Delta \tau \right) = \frac{\sin \left( \sqrt{K} \, \frac{1 - \Delta}2 \right)}{\cos \left( \sqrt{K} \, \frac{2\delta + \Delta}2 \right)} + \frac{\sqrt{K} \delta\cos \left( \sqrt{K} \, \frac{2\delta + \Delta}2 \right)}{\cos \left( \sqrt{K} \, \frac{1 - \Delta}2 \right)}
\end{cases} \, ,
\end{equation}
and let us define the new variables
\begin{equation}
\begin{cases}
x = \sqrt{K} \, \frac{2\delta + \Delta}2 \\
y = \frac{\pi}2 - \sqrt{K} \, \frac{1 - \Delta}2 \\
z = \alpha\sqrt{K} \Delta\tau
\end{cases} \, .
\end{equation}
From \eqref{nonlinearsys2} we have
\begin{equation} \label{nonlinearsysxyz}
\begin{cases}
\alpha = \frac{\sinh z \cos x}{\sin y} \\
\cosh z = \frac{\cos y}{\cos x} + \frac{\sqrt{K} \delta\cos x}{\sin y}
\end{cases} \, .
\end{equation}
Since
\begin{equation}
\sqrt{K} \delta = x - y + \frac{\pi}2 - \frac{\sqrt{K}}2 \, ,
\end{equation}
we obtain from \eqref{nonlinearsysxyz}
\begin{equation}
\frac{\sqrt{K}}2 = \frac{\pi}2 + x - y - \sin y \left( \frac{\cosh z}{\cos x} - \frac{\cos y}{\cos^2 x} \right) \, .
\end{equation}
To prove blowup with $K<\pi^2$ we look for values of $(x,y,z)$ close to $(0,0,0)$ such that
\begin{equation}
\begin{cases}
\sqrt{K} < \pi \\
\Delta > 0 \\
\delta > 0 \\
\Delta \tau < \frac12 \Delta
\end{cases} 
\Leftrightarrow
\begin{cases}
x - y - \sin y \left( \frac{\cosh z}{\cos x} - \frac{\cos y}{\cos^2 x} \right) < 0 \\
y - \frac{\pi}2 + \frac{\sqrt{K}}2 > 0 \\
x - y + \frac{\pi}2 - \frac{\sqrt{K}}2 > 0 \\
\frac{z}{\alpha} < y - \frac{\pi}2 + \frac{\sqrt{K}}2
\end{cases} 
\, ,
\end{equation}
that is,
\begin{equation} \label{ineqs}
\begin{cases}
x - y - \sin y \left( \frac{\cosh z}{\cos x} - \frac{\cos y}{\cos^2 x} \right) < 0 \\
x - \sin y \left( \frac{\cosh z}{\cos x} - \frac{\cos y}{\cos^2 x} \right) > 0 \\
\sin y \left( \frac{\cosh z}{\cos x} - \frac{\cos y}{\cos^2 x} \right) > 0 \\
\frac{\sin y}{\cos x}\frac{z}{\sinh z} < x - \sin y \left( \frac{\cosh z}{\cos x} - \frac{\cos y}{\cos^2 x} \right)
\end{cases} 
\, .
\end{equation}
We will look for solutions of these conditions in the first octant, that is, we will assume {\em a priori} that $x,y,z >0$ (and also that they are sufficiently small). For such values of $(x,y,z)$, the last condition implies the second condition. The Taylor expansions of the expressions in the remaining inequalities read
\begin{equation} \label{ineqs}
\begin{cases}
x < y \left( 1 - \frac{x^2}2 + \frac{y^2}2 + \frac{z^2}2 \right) + O(r^4) \\
- \frac{x^2}2 + \frac{y^2}2 + \frac{z^2}2 > O(r^3) \\
x > y \left( 1 + \frac{y^2}3 + \frac{z^2}3 \right) + O(r^4) 
\end{cases} 
\, ,
\end{equation}
where $r^2 = x^2 + y^2 + z^2$. We can therefore obtain solutions of \eqref{ineqs} by taking
\begin{equation}
(x,y,z) = (t + 6 t^3, t, 3t)
\end{equation}
for $t>0$ sufficiently small.
\end{proof}

\begin{Remark}
Numerically, the solution to the IBVP \eqref{IBVP} for constant $K$ appears to blow up for $K \gtrsim 1.5$, with the blow-up occurring in region $\mathcal{B}$ for $K \gtrsim 4.4$.
\end{Remark}

\begin{Remark}
It is interesting to note that the linearized version of the IBVP~\ref{IBVP}, featuring a non-homogeneous Klein-Gordon equation,
\begin{equation}
\begin{cases} \label{KleinGordon}
- \ddot{u} + u'' = K (1- 2u) \\
u(0,\lambda)=\dot{u}(0,\lambda)=0, \quad \lambda \in [0,1] \\
u(\tau,0) = u(\tau,1) = 0, \quad \tau\geq 0
\end{cases} \, ,
\end{equation}
can be written explicitly in the constant curvature case as
\begin{equation}
u(\tau,\lambda) = u_0(\lambda) + \sum_{n=1}^{+\infty} c_n(\tau) \sin(n \pi \lambda) \, ,
\end{equation}
for appropriate functions $u_0(\lambda)$ and $c_n(\tau)$ satisfying
\begin{equation}
u_0'' + 2Ku_0 = K
\end{equation}
and
\begin{equation}
\ddot{c}_n + (n^2 \pi^2 - 2K) c_n = 0 \, .
\end{equation}
Therefore, the (global in time, quasiperiodic) solutions of the IBVP~\eqref{KleinGordon} are bounded in time for $K < \frac{\pi^2}2$, and unbounded for $K \geq \frac{\pi^2}2$.
\end{Remark}
%
%
\section{The tethered galaxy problem}\label{section5}
We now consider the tethered galaxy problem, where the rod is used to connect two galaxies, modelled as point masses with rest mass $m>0$ (note that this is not exactly the same problem as that discussed in \cite{DLW03}, where a tether initially keeps the galaxies at a fixed spatial distance and then is cut). This can be seen as the relativistic analogue of a system of two masses connected by a spring.

The proper acceleration of the galaxies is (see Appendix~\ref{appendixA})
\begin{equation}
A = - u' e^{2u} \frac{\partial}{\partial \lambda} \, ,
\end{equation}
while the force on the galaxy at $\lambda=1$ is
\begin{equation}
F = p \, e^u \frac{\partial}{\partial \lambda} = \frac{\mu_0}2 \left( e^{2u} - 1 \right) e^u \frac{\partial}{\partial \lambda} \, .
\end{equation}
Therefore, at $\lambda=1$ we obtain the nonlinear boundary condition
\begin{equation}
F = m A \Leftrightarrow u' = - \frac{\mu_0}{m} \sinh u \, .
\end{equation}
At $\lambda=0$ the force is reversed, and so we obtain the same boundary condition but with a plus sign. If again we assume that the rod is initially unstretched and not moving, we are led to the following IBVP:
\begin{equation} \label{IBVP2}
\begin{cases}
- \ddot{u} + u'' = K e^{- 2u} \\
u(0,\lambda)=\dot{u}(0,\lambda)=0, \quad \quad \lambda \in [0,1] \\
u'(\tau,0) = \sigma \sinh u(\tau,0), \quad \,\,\,\, \tau\geq 0 \\
u'(\tau,1) = - \sigma \sinh u(\tau,1), \quad \tau\geq 0
\end{cases} \, ,
\end{equation}
where\footnote{Note that the IBVP~\eqref{IBVP} studied in the previous sections corresponds to the limit of this IBVP when $m \to 0$, that is, when $\sigma \to + \infty$.}
\begin{equation}
\sigma = \frac{\mu_0}{m} \, .
\end{equation}

Consider the modified energy
\begin{equation}
H(\tau) = \int_0^1 \left( \dot{u}^2 + u'^2 - K e^{-2u} \right) d\lambda + 2\sigma \Big[ \cosh u(\tau,0) + \cosh u(\tau,1) \Big] \, .
\end{equation}
Still under the assumption $\dot{K} \geq 0$, we have (recall \eqref{energycalculation})
\begin{align}
\dot{H}(\tau) & = \Big[ 2\dot{u}u' \Big]_0^1 - \int_0^1 \dot{K} e^{-2u}d\lambda + 2\sigma \Big[ \dot{u}(\tau,0)\sinh u(\tau,0) + \dot{u}(\tau,1)\sinh u(\tau,1) \Big] \nonumber \\ 
& = - \int_0^1 \dot{K} e^{-2u}d\lambda \leq 0 \, , 
\end{align}
that is, $H(\tau)$ is decreasing. This allows us to prove the following result.

\begin{Thm} \label{globalintime2}
If $\dot{K} \geq 0$ and either $K \leq 0$ or $0 \leq K \leq \kappa$, for sufficiently small $\kappa > 0$, then the solution to the IBVP \eqref{IBVP2} is global in time towards the future and bounded.
\end{Thm}

\begin{proof}
According to the blow-up alternative in Theorem~\ref{GWPBA} (see Appendix~\ref{appendixB}), if $\|u(\tau,\cdot)\|_{L^\infty}$ remains bounded for $\tau>0$ then the unique solution to the IBVP can be extended globally in time towards the future. Since
\begin{equation}
u(\tau,\lambda) = u(\tau,0) + \int_0^\lambda u'(\tau,s) ds \, ,
\end{equation}
we have
\begin{equation}
\|u(\tau,\cdot)\|_{L^\infty} \leq |u(\tau,0)| + \|u'(\tau,\cdot)\|_{L^1} \leq \cosh u(\tau,0) + \|u'(\tau,\cdot)\|_{L^2} \, ,
\end{equation}
where we used the Cauchy-Schwarz inequality and $|x|<\cosh x$. The same inequality with $\cosh u(\tau,1)$ replacing $\cosh u(\tau,0)$ is obviously true, and so
\begin{equation}
\|u(\tau,\cdot)\|_{L^\infty} \leq \frac12 \Big[ \cosh u(\tau,0) + \cosh u(\tau,1) \Big] + \|u'(\tau,\cdot)\|_{L^2} \, .
\end{equation}
Consequently, we have
\begin{equation}
\|u\|_{L^\infty} \leq \frac1{4\sigma} \left( H(\tau) + \int_0^1 K e^{-2u} d\lambda \right) +  \left( H(\tau) + \int_0^1 K e^{-2u} d\lambda \right)^\frac12 \, .
\end{equation}
Since the function
\begin{equation}
f(x)=\frac1{4\sigma}x + x^\frac12
\end{equation}
is clearly increasing for $x \geq 0$, we have
\begin{equation}
\|u\|_{L^\infty} \leq \frac1{4\sigma} \left( H(0) + \int_0^1 K e^{-2u} d\lambda \right) +  \left( H(0) + \int_0^1 K e^{-2u} d\lambda \right)^\frac12 \, .
\end{equation}
If $K \leq 0$, we have
\begin{equation}
\|u\|^2_{L^\infty} \leq \frac1{4\sigma} H(0) +  H(0)^\frac12 \, ,
\end{equation}
and so $u$ is bounded. If $0 \leq K \leq \kappa$, we have
\begin{equation}
H(0) = - \int_0^1 K e^{-2u} d\lambda + 4 \sigma \leq 4\sigma \, ,
\end{equation}
and so
\begin{align}
\|u\|_{L^\infty} & \leq \frac1{4\sigma} \left( 4\sigma + \int_0^1 K e^{-2u} d\lambda \right) +  \left( 4\sigma + \int_0^1 K e^{-2u} d\lambda \right)^\frac12 \nonumber \\
& \leq \frac1{4\sigma} \left( 4\sigma + \kappa e^{2 \|u\|_{L^\infty}} \right) +  \left( 4\sigma + \kappa e^{2 \|u\|_{L^\infty}} \right)^\frac12 \, . \label{ineq2}
\end{align}
Consider the function
\begin{equation}
g(x) = x - \frac1{4\sigma} \left( 4\sigma + \kappa e^{2x} \right) -  \left( 4\sigma + \kappa e^{2x} \right)^\frac12 \, ,
\end{equation}
which satisfies
\begin{equation}
g(0)= - \frac1{4\sigma} \left( 4\sigma + \kappa \right) -  \left( 4\sigma + \kappa \right)^\frac12 < 0 \, .
\end{equation}
Given
\begin{equation}
x_0 > 2\sqrt{\sigma} + 1 \, ,
\end{equation}
we have
\begin{equation}
\lim_{\kappa \to 0} g(x_0) = x_0 - 1 - 2\sqrt{\sigma} > 0 \, ,
\end{equation}
and so $g(x_0)>0$ for sufficiently small $\kappa > 0$. Since from \eqref{ineq2} we have
\begin{equation}
g(\|u\|_{L^\infty}) \leq 0 \, ,
\end{equation}
and given that $\|u(0,\cdot)\|_{L^\infty}=0$, we conclude that, for $\kappa > 0$ sufficiently small, $\|u(\tau,\cdot)\|_{L^\infty}<x_0$ for $\tau>0$.
\end{proof}

\begin{Remark}
Note that the result in Theorem~\ref{globalintime2} is weaker than the corresponding result in Theorem~\ref{globalintime}. Physically, this difference arises because it is easier for the expansion to stretch the rod when there are masses attached to its endpoints, as the tidal forces on these masses add to the body forces stretching the rod.
\end{Remark}

Since the blow-up argument in Theorem~\ref{blowupThm} is completely independent of the boundary conditions, we also have the following result.

\begin{Thm}
If $\dot{K} \geq 0$ and $K \geq \kappa > \pi^2$ then the solution to the IBVP \eqref{IBVP2} blows up before $\tau=1/2$.
\end{Thm} 

In the constant curvature case, we also have the following result:

\begin{Thm}
If $K>0$ is constant then the solution to the IBVP~\eqref{IBVP2} blows up for some $K < \pi^2$.
\end{Thm}

\begin{proof}
The proof of Theorem~\ref{constantcurvature} applies in this case if we can establish that $u(\tau,0)\leq 0$ for $\tau \in [0,1]$. Now from the expression of the exact solution in region $\mathcal{D}$ we have $L^+u \leq 0$ in this region, and by integrating $L^-L^+ u = - K e^{-2u}$ from the boundary of $\mathcal{D}$ into the region $\mathcal{B}$ we obtain $L^+ u \leq 0$ in this region. Therefore we have
\begin{equation}
\dot{u}(\tau,0) + \sigma \sinh u(\tau,0) = \dot{u}(\tau,0) + u'(\tau,0) = L^+(\tau,0) \leq 0
\end{equation}
for $\tau \in [0,1]$. Since $u(0,0)=0$, we conclude from the usual comparison theorem for ODEs that $u(\tau,0)\leq z(\tau)$, where $z(\tau)$ is the solution to the ODE Cauchy problem
\begin{equation}
\begin{cases}
\dot{z} = - \sigma\sinh z \\
z(0)=0
\end{cases} \, ,
\end{equation}
that is, $z(\tau) \equiv 0$.
\end{proof}

In other words, the qualitative behavior of the tethered galaxy system in response to the cosmological expansion is exactly the same as that of the rod alone. Taken together, these two examples provide mathematically rigorous support for the widely held view that sufficiently small bound systems do not follow the Hubble flow, whereas larger systems may do so.
%
%
\section*{Acknowledgements}
This work was partially supported by FCT/Portugal through CAMGSD, IST-ID (projects UIDB/04459/2020 and UIDP/04459/2020) and project 2024.04456.CERN, and also by the H2020-MSCA-2022-SE project EinsteinWaves, GA no.~101131233. 
%
%
\appendix
\section{Geometry calculations}\label{appendixA}
The metric \eqref{FLRWmetric2D} can be written as
\begin{equation}
ds^2 = \left(\omega^0\right)^2 + \left(\omega^1\right)^2 \quad \text{with} \quad \omega^0 = dt \, , \quad \omega^1 = a(t) dx \, .
\end{equation}
From
\begin{equation}
\begin{cases}
d\omega^0 = 0 = - \omega^0_{\,\,\,\,1} \wedge \omega^1 \\
d\omega^1 = \dot{a} dt \wedge dx = - \omega^1_{\,\,\,\,0} \wedge \omega^0
\end{cases}
\end{equation}
we obtain the connection form
\begin{equation}
\omega^0_{\,\,\,\,1} = \omega^1_{\,\,\,\,0} = \dot{a} \, dx \, ,
\end{equation}
and consequently the curvature form
\begin{equation}
\Omega^0_{\,\,\,\,1} = d\omega^0_{\,\,\,\,1} = \ddot{a} \, dt \wedge dx = \frac{\ddot{a}}{a} \, \omega^0 \wedge \omega^1\, ,
\end{equation}
so that the Gauss curvature is
\begin{equation} \label{Gausscurvature}
K =  \frac{\ddot{a}}{a} \, .
\end{equation}
The metric \eqref{proper}, on the other hand, can be written as
\begin{equation}
ds^2 = \left(\omega^0\right)^2 + \left(\omega^1\right)^2 \quad \text{with} \quad \omega^0 = e^{-u} d\tau \, , \quad \omega^1 = e^{-u} d\lambda \, .
\end{equation}
From
\begin{equation}
\begin{cases}
d\omega^0 = - u' e^{-u} d\lambda \wedge d\tau = - \omega^0_{\,\,\,\,1} \wedge \omega^1 \\
d\omega^1 = - \dot{u} e^{-u} d\tau \wedge d\lambda = - \omega^1_{\,\,\,\,0} \wedge \omega^0
\end{cases}
\end{equation}
we obtain the connection form
\begin{equation}
\omega^0_{\,\,\,\,1} = \omega^1_{\,\,\,\,0} = - u' d\tau - \dot{u} \, d\lambda \, ,
\end{equation}
and consequently the curvature form
\begin{equation}
\Omega^0_{\,\,\,\,1} = d\omega^0_{\,\,\,\,1} = - u'' d\lambda \wedge d\tau - \ddot{u} \, d\tau \wedge d\lambda = e^{2u} (-\ddot{u} + u'') \, \omega^0 \wedge \omega^1\, ,
\end{equation}
so that the Gauss curvature is
\begin{equation} \label{Gausscurvature}
K = e^{2u} (-\ddot{u} + u'') \, .
\end{equation}
Finally, note that if
\begin{equation}
E_0 = e^{u} \frac{\partial}{\partial \tau} \, , \qquad E_1 = e^{u} \frac{\partial}{\partial \lambda}
\end{equation}
is the orthonormal frame dual to $\{\omega^0,\omega^1\}$ then the covariant acceleration of the observers at rest with respect to the rod is
\begin{equation}
\nabla_{E_0} E_0 =  \omega^1_{\,\,\,\,0}(E_0) E_1 = - u' e^{u} E_1 \, .
\end{equation}
\section{Existence and uniqueness results}\label{appendixB}
In this appendix we prove an existence and uniqueness result for classical solutions of the IBVPs considered in this work.

\begin{Thm} \label{GWPBA}
The IBVPs~\eqref{IBVP} and \eqref{IBVP2} admit unique $C^2$ solutions on maximal time intervals $[0,T)$, where either $T = + \infty$ or $\limsup_{t {\scriptscriptstyle \nearrow} T} \| u(\tau, \cdot)\|_{L^\infty} = + \infty$.
\end{Thm}

\begin{proof}
We present only the proof for the more complicated case of the IBVP~\eqref{IBVP2}; the proof for the simpler IBVP~\eqref{IBVP} is a trivial adaptation.

Let us consider first the IBVP
\begin{equation} \label{IBVP3}
\begin{cases}
- \ddot{u} + u'' = K e^{- 2u} \\
u(\tau_0,\lambda)=f(\lambda), \quad \quad \lambda \in [0,1] \\
\dot{u}(\tau_0,\lambda)=g(\lambda), \quad \quad \lambda \in [0,1] \\
u'(\tau,0) = \sigma \sinh u(\tau,0), \quad \,\,\,\, \tau\in [\tau_0, \tau_1] \\
u'(\tau,1) = - \sigma \sinh u(\tau,1), \quad \tau\in [\tau_0, \tau_1]
\end{cases} \, ,
\end{equation}
where we assume that $\Delta \tau = \tau_1 - \tau_0 < \frac12$ and that $f \in C^2([0,1])$ and $g \in C^1([0,1])$ satisfy the compatibility conditions
\begin{equation} \label{compatibility1}
\begin{cases}
f'(0)=\sigma \sinh f(0) \\
f'(1)=-\sigma \sinh f(1)
\end{cases}
\end{equation}
and
\begin{equation} \label{compatibility2}
\begin{cases}
g'(0)=\sigma g(0) \cosh f(0) \\
g'(1)=-\sigma g(1) \cosh f(1)
\end{cases} \, .
\end{equation}
To solve this problem, we introduce the new variables
\begin{equation}
v = L^+ u \, , \qquad w = L^- u \, ,
\end{equation}
where $L^+$ and $L^-$ are the vector fields defined in \eqref{L+-}, and study the first order IBVP
\begin{equation} \label{IBVPfirstorder}
\begin{cases}
\dot{u} = \frac12 (v + w) \\
L^- v = - K e^{- 2u} \\
L^+ w = - K e^{- 2u} \\
u(\tau_0,\lambda)=f(\lambda), \quad \quad \lambda \in [0,1] \\
v(\tau_0,\lambda)=g(\lambda) + f'(\lambda), \quad \quad \lambda \in [0,1] \\
w(\tau_0,\lambda)=g(\lambda) - f'(\lambda), \quad \quad \lambda \in [0,1] \\
v(\tau,0) - w(\tau,0) = 2\sigma \sinh u(\tau,0), \quad \,\,\,\, \tau\in [\tau_0, \tau_1] \\
w(\tau,1) - v(\tau,1) = 2\sigma \sinh u(\tau,1), \quad \tau\in [\tau_0, \tau_1]
\end{cases} \, ,
\end{equation}which we write as the following system of integral equations:
\begin{equation} \label{IBVP4}
\begin{cases}
u(\tau,\lambda) = f(\lambda) + \frac12 \int_{\tau_0}^\tau \left[ v(s,\lambda) + w(s,\lambda) \right] ds \\
v(\tau,\lambda)=g(\lambda + \tau - \tau_0) + f'(\lambda + \tau - \tau_0) - \int_{\tau_0}^\tau K(s,\lambda+\tau-s) e^{- 2u(s,\lambda+\tau-s)} ds \, , \quad \tau + \lambda \leq 1 + \tau_0 \\
v(\tau,\lambda)= w_\rig(\tau, \lambda) - 2 \sigma \sinh u_\rig(\tau,\lambda) - \int_{\tau + \lambda - 1}^\tau K(s,\lambda+\tau-s) e^{- 2u(s,\lambda+\tau-s)} ds \, , \quad \tau + \lambda\geq 1 + \tau_0 \\
w(\tau,\lambda)=g(\lambda - \tau + \tau_0) - f'(\lambda - \tau + \tau_0) - \int_{\tau_0}^\tau K(s,\lambda-\tau+s) e^{- 2u(s,\lambda-\tau+s)} ds \, , \quad \tau - \lambda \leq \tau_0 \\
w(\tau,\lambda)=v_\lef(\tau,\lambda) - 2 \sigma \sinh u_\lef(\tau,\lambda) - \int_{\tau - \lambda}^\tau K(s,\lambda-\tau+s) e^{- 2u(s,\lambda-\tau+s)} ds \, , \quad \tau - \lambda \geq \tau_0 \\
u_\lef(\tau,\lambda) = f(0) + \frac12 \int_{\tau_0}^{\tau-\lambda} \left[ v(s,0) + w(s,0) \right] ds \\
u_\rig(\tau,\lambda) = f(1) + \frac12 \int_{\tau_0}^{\tau+\lambda-1} \left[ v(s,1) + w(s,1) \right] ds \\
v_\lef(\tau,\lambda)=g(\tau - \lambda - \tau_0) + f'(\tau - \lambda - \tau_0) - \int_{\tau_0}^{\tau-\lambda} K(s,\tau-\lambda-s) e^{- 2u(s,\tau-\lambda-s)} ds \\
w_\rig(\tau,\lambda)=g(2 - \tau - \lambda + \tau_0) - f'(2 - \tau - \lambda + \tau_0) - \int_{\tau_0}^{\tau+\lambda-1} K(s,2 - \tau - \lambda+s) e^{- 2u(s,2 - \tau - \lambda+s)} ds
\end{cases} \, .
\end{equation}

Here the first equation states that $u$ is the time integral of its time derivative, and the next four equations give $u$ and $v$ by integrating the equation along the characteristics, starting either at the initial time or at the boundaries; the functions $u_\lef$, $u_\rig$, $v_\lef$ and $w_\rig$ are the boundary values needed to start the integration at the boundaries, and are themselves obtained by integrations that always start at the initial time (since $\Delta \tau < \frac12$).

Consider the Banach space $X= \left(C^0\left([\tau_0, \tau_1] \times [0,1]\right)\right)^3$ with the norm
\begin{equation}
\|(u,v,w)\| = \| u \|_{L^\infty} + \| v \|_{L^\infty} + \| w \|_{L^\infty} \, ,
\end{equation}
and let $\cF:X \to X$ be the obvious operator implicit in~\eqref{IBVP4}, whose fixed points are the solutions of \eqref{IBVPfirstorder} (note that the compatibility conditions~\eqref{compatibility1} ensure that $\cF$ indeed maps continuous functions to continuous functions). If we write
\begin{equation}
\cF(u,v,w) = (\bar{u},\bar{v},\bar{w}) \, ,
\end{equation}
and assume that
\begin{equation}
\|(u,v,w)\| \leq R \, ,
\end{equation}
for a suitable constant $R>0$, then we have
\begin{equation}
\| \bar{u} \|_{L^\infty} \leq \| f \|_{L^\infty} + R \Delta \tau \, .
\end{equation}
Similarly, we have, for $\tau + \lambda \leq 1 + \tau_0$,
\begin{equation}
| \bar{v} | \leq \| g \|_{L^\infty} + \| f' \|_{L^\infty} + \| K \|_{L^\infty} e^{2R} \Delta \tau \, ,
\end{equation}
and, for $ \tau + \lambda\geq 1 + \tau_0$,
\begin{equation}
| \bar{v} | \leq \| w_\rig \|_{L^\infty} + 2 \sigma \| \sinh u_\rig \|_{L^\infty} + \| K \|_{L^\infty} e^{2R} \Delta \tau \, .
\end{equation}
Since
\begin{equation}
\| w_\rig \|_{L^\infty} \leq \| g \|_{L^\infty} + \| f' \|_{L^\infty} + \| K \|_{L^\infty} e^{2R} \Delta \tau \, ,
\end{equation}
and
\begin{equation}
\| u_\rig \|_{L^\infty} \leq \| f \|_{L^\infty} + R \Delta \tau \, ,
\end{equation}
we conclude that
\begin{equation}
\| \bar{v} \|_{L^\infty} \leq \| g \|_{L^\infty} + \| f' \|_{L^\infty} + 2 \| K \|_{L^\infty} e^{2R} \Delta \tau + 2 \sigma \sinh \left( \| f \|_{L^\infty} + R \Delta \tau \right)  \, ,
\end{equation}
and similarly
\begin{equation}
\| \bar{w} \|_{L^\infty} \leq \| g \|_{L^\infty} + \| f' \|_{L^\infty} + 2 \| K \|_{L^\infty} e^{2R} \Delta \tau + 2 \sigma \sinh \left( \| f \|_{L^\infty} + R \Delta \tau \right)  \, .
\end{equation}
We then have
\begin{equation}
\| (\bar{u},\bar{v},\bar{w}) \| \leq  \| f \|_{L^\infty} + 2\| g \|_{L^\infty} + 2\| f' \|_{L^\infty} + R \Delta \tau  + 4 \| K \|_{L^\infty} e^{2R} \Delta \tau + 4 \sigma \sinh \left( \| f \|_{L^\infty} + R \Delta \tau \right)  \, ,
\end{equation}
and so if we choose
\begin{equation} \label{sizeofR}
R =  \| f \|_{L^\infty} + 2\| g \|_{L^\infty} + 2\| f' \|_{L^\infty} + 4 \sigma \sinh \left( \| f \|_{L^\infty} \right) + 1 \, ,
\end{equation}
and then $\Delta \tau$ sufficiently small so that
\begin{equation} \label{sizeofdeltatau}
R \Delta \tau  + 4 \| K \|_{L^\infty} e^{2R} \Delta \tau + 4 \sigma \sinh \left( \| f \|_{L^\infty} + R \Delta \tau \right) - 4 \sigma \sinh \left( \| f \|_{L^\infty} \right) \leq 1 \, ,
\end{equation}
we have
\begin{equation}
\| \cF(u,v,w) \| \leq  R  \, ,
\end{equation}
that is, $\cF$ maps $\overline{B_R(0)}$ to itself.

Let us take $(u_1,v_1,w_1)$ and $(u_2,v_2,w_2)$ in $\overline{B_R(0)}$. We have
\begin{equation}
\| \bar{u}_1 - \bar{u}_2 \|_{L^\infty} \leq \Delta \tau \| v_1 - v_2 \|_{L^\infty} + \Delta \tau \| w_1 - w_2 \|_{L^\infty}
\end{equation}
and, for $\tau + \lambda \leq 1 + \tau_0$,
\begin{equation}
| \bar{v}_1 - \bar{v}_2 | \leq 2 \Delta \tau  \| K \|_{L^\infty} e^{2R} \| u_1 - u_2 \|_{L^\infty} \, .
\end{equation}
For $ \tau + \lambda\geq 1 + \tau_0$, on the other hand, we have
\begin{equation}
| \bar{v}_1 - \bar{v}_2 | \leq \| w_{\rig,1} - w_{\rig,2} \|_{L^\infty} + 2 \sigma \cosh R \| u_{\rig,1} - u_{\rig,2} \|_{L^\infty} + 2 \Delta \tau  \| K \|_{L^\infty} e^{2R} \| u_1 - u_2 \|_{L^\infty} \, .
\end{equation}
Since
\begin{equation}
\| w_{\rig,1} - w_{\rig,2} \|_{L^\infty} \leq 2 \Delta \tau  \| K \|_{L^\infty} e^{2R} \| u_1 - u_2 \|_{L^\infty} \, ,
\end{equation}
and
\begin{equation}
\| u_{\rig,1} - u_{\rig,2} \|_{L^\infty} \leq \Delta \tau \| v_1 - v_2 \|_{L^\infty} + \Delta \tau \| w_1 - w_2 \|_{L^\infty} \, ,
\end{equation}
we conclude that
\begin{equation}
\| \bar{v}_1 - \bar{v}_2 \|_{L^\infty} \leq 4 \Delta \tau  \| K \|_{L^\infty} e^{2R} \| u_1 - u_2 \|_{L^\infty} + 2\sigma \Delta \tau \cosh R \| v_1 - v_2 \|_{L^\infty} + 2\sigma \Delta \tau \cosh R \| w_1 - w_2 \|_{L^\infty}  \, ,
\end{equation}
and similarly
\begin{equation}
\| \bar{w}_1 - \bar{w}_2 \|_{L^\infty} \leq 4 \Delta \tau  \| K \|_{L^\infty} e^{2R} \| u_1 - u_2 \|_{L^\infty} + 2\sigma \Delta \tau \cosh R \| v_1 - v_2 \|_{L^\infty} + 2\sigma \Delta \tau \cosh R \| w_1 - w_2 \|_{L^\infty}  \, .
\end{equation}
In other words, we have
\begin{equation}
\| (\bar{u}_1,\bar{v}_1,\bar{w}_1) - (\bar{u}_2,\bar{v}_2,\bar{w}_2) \| \leq  M \Delta \tau \| (u_1, v_1, w_1) - (u_2, v_2, w_2) \|  \, ,
\end{equation}
where
\begin{equation} \label{sizeofM}
M = \max \left\{ 8 \| K \|_{L^\infty} e^{2R}, 1 + 4 \sigma \cosh R \right\}  \, .
\end{equation}
We conclude that $\cF$ is a contraction if $\Delta \tau$ is chosen small enough. Note from \eqref{sizeofR}, \eqref{sizeofdeltatau} and \eqref{sizeofM} that the choice of $\Delta \tau$ depends only on $\| f \|_{L^\infty}$, $\| g \|_{L^\infty}$, $\| f' \|_{L^\infty}$ and $\| K \|_{L^\infty}$.

The fixed point of $\cF$ yields a solution of the system of integral equations~\eqref{IBVP4} in the space of continuous functions. Since these equations imply
\begin{equation}
L^- v = L^+w  \, ,
\end{equation}
there exists a $C^1$ function $\tilde{u}$ such that\footnote{Recall that if $P(x,y)$ and $Q(x,y)$ are continuous and $\frac{\partial P}{\partial y} = \frac{\partial Q}{\partial x}$, where these partial derivatives exist and are continuous, then $\phi(x,y) = \int_0^x P(s,0) ds + \int_0^y Q(x,s) ds$ is a $C^1$ potential for $(P,Q)$.}
\begin{equation}
L^+\tilde{u} = v \, , \qquad L^-\tilde{u} = w \, ,
\end{equation}
and consequently
\begin{equation}
\dot{\tilde{u}} = L^+\tilde{u} + L^-\tilde{u} = v + w = \dot{u} \, ,
\end{equation}
that is,
\begin{equation} \label{uutilde}
\tilde{u}(\tau,\lambda) = u(\tau,\lambda) + h(\lambda)
\end{equation}
for some $C^1$ function $h$ (recall that $u(0,\lambda)=f(\lambda)$ is $C^2$). Since from~\eqref{IBVP4} we have
\begin{equation}
L^+u(\tau_0,\lambda) = f'(\lambda) + \frac12 v(\tau_0,\lambda) + \frac12 w(\tau_0,\lambda) = f'(\lambda) + g(\lambda) = v(\tau_0,\lambda) = L^+ \tilde{u}(\tau_0,\lambda) \, ,
\end{equation}
we obtain from \eqref{uutilde} that $h'(\lambda)=0$, that is, $u$ and $\tilde{u}$ differ by a constant, and so
\begin{equation}
L^+u = v \, , \qquad L^-u = w \, .
\end{equation}
Consequently, $u$ is in fact $C^1$, and so our hypotheses on the regularity of the initial data $f$ and $g$, together with \eqref{IBVP4}, show that $v$ and $w$ are also of class $C^1$ (note that this is not obvious along the curves $\tau + \lambda = 1 + \tau_0$ and $\tau - \lambda = \tau_0$, where the compatibility condition~\eqref{compatibility2} must be used). From this it is easily shown that $u$ is a $C^2$ solution of the IBVP~\eqref{IBVP3}.

Going back to the original IBVP~\eqref{IBVP2}, we see that its initial data satisfies the compatibility conditions~\eqref{compatibility1} and \eqref{compatibility2}. Consequently, it has a $C^2$ solution for $\tau \in [0, \Delta \tau]$. If we subsequently set $\tau_0 = \Delta \tau$, $f(\lambda)=u(\tau_0,\lambda)$ and $g(\lambda)=\dot{u}(\tau_0,\lambda)$, then $f$ and $g$ satisfy the compatibility conditions~\eqref{compatibility1} and \eqref{compatibility2}, and so the IBVP~\eqref{IBVP3} has a $C^2$ solution for $\tau \in [\tau_0, \tau_0 + \Delta \tau']$, meaning that the original IBVP~\eqref{IBVP2} has a $C^2$ solution for $\tau \in [0, \Delta \tau + \Delta \tau']$. Iterating this procedure, we can keep extending the solution to larger and larger time intervals. Since the time interval by which we can extend depends only on $\| f \|_{L^\infty}$, $\| g \|_{L^\infty}$, $\| f' \|_{L^\infty}$ and $\| K \|_{L^\infty}$, and the latter is smooth for $\tau \in [0, +\infty)$, we see that if $\| u(\tau, \cdot) \|_{L^\infty}$, $\| \dot{u}(\tau, \cdot) \|_{L^\infty}$ and $\| u'(\tau, \cdot) \|_{L^\infty}$ remain bounded then the solution is global in time. On the other hand, it is clear from the system of integral equations~\eqref{IBVP4} that $\| \dot{u}(\tau, \cdot) \|_{L^\infty}$ and $\| u'(\tau, \cdot) \|_{L^\infty}$ canot blow up if $\| u(\tau, \cdot) \|_{L^\infty}$ remains bounded. We conclude that either the solution is global or $\| u(\tau, \cdot) \|_{L^\infty}$ blows up.

Finally, we show that the solution is unique: given two $C^2$ solutions $u_1$ and $u_2$ defined on some compact time interval $[0, \tau_1]$, we take $R$ large enough such that both $(u_1,L^+{u}_1,L^-u_1)$ and $(u_2,L^+u_2,L^-u_2)$ are in the closed ball $\overline{B_R(0)}$ of the space $\left(C^0\left([0, \tau_1] \times [0,1]\right)\right)^3$, and choose $\Delta \tau = \tau_1 / N$ (for some $N \in \bbN$) small enough such that $\cF$ is a contraction on the closed ball $\overline{B_R(0)}$ of the space $\left(C^0\left([0,\Delta \tau] \times [0,1]\right)\right)^3$. Since both $(u_1,L^+{u}_1,L^-u_1)$ and $(u_2,L^+u_2,L^-u_2)$ are fixed points of $\cF$, they must coincide for $\tau \in [0,\Delta \tau]$. Iterating this procedure for time intervals of the form $[ k \Delta \tau, (k + 1) \Delta\tau]$, with $k = 1, \ldots, N-1$, we conclude that $u_1$ and $u_2$ coincide on $[0, \tau_1]$. 
\end{proof}

\begin{Remark}
The existence and uniqueness result for the IBVP~\eqref{IBVP}, under weaker regularity assumptions (namely $H^1 \times L^2$), can be obtained from the general result in \cite[Theorem~6.2.2]{CH98}; to the best of our knowledge, a similar result for nonlinear boundary conditions, applicable to the IBVP~\eqref{IBVP2}, is not available in the literature.
\end{Remark}
%
%
%

\end{document}